\documentclass[11pt,runningheads]{llncs}
\usepackage{enumitem}
\usepackage{amsmath}
\usepackage{color}
\usepackage{soul}
\usepackage{graphicx}

\title{Online Facility Assignments on Polygons}

\author{Sumaiya Malik$^1$, Reyan Ahmed$^2$, Md. Manzurul Hasan$^1$}
\institute{Department of Computer Science, American International University-Bangladesh, Dhaka 1229, Bangladesh$^1$\\ Department of Computer Science, University of Arizona, United States of America$^2$\\}

\begin{document}

\maketitle

%\author{Sumaiya Malik}
%\email{sumaiyamalik2000@gmail.com}
%\affiliation{Department of Computer Science, American International University-Bangladesh, Dhaka 1229, Bangladesh}
%\author{Abu Reyan Ahmed\textsuperscript{*,}}
%\email{abureyanahmed@email.arizona.edu}
%\affiliation{Department of Computer Science, University of Arizona, United States of America}
%\author{Md. Manzurul Hasan}
%\email{manzurul@aiub.edu}
%\affiliation{Department of Computer Science, American International University-Bangladesh, Dhaka 1229, Bangladesh}
%\email{abureyanahmed@email.arizona.edu}
%\shortauthors{K. Velan}

%\received{00 January 2009}
%\revised{00 Month 2009}

%\category{C.2}{Computer Communication Networks}{Computer Networks}
%\category{C.4}{Performance of Systems}{Analytical Models}
%\category{G.3}{Stochastic Processes}{Queueing Systems}
%\terms{Internet Technologies, E-Commerce}
%\keywords{Online Assignment; Greedy Algorithm; Multiple Customer Assignments; Geometrical Shapes}

\begin{abstract}
%{\color{red}The polygons we are considering have equal side lengths. How can we point that out?}
We study the online facility assignment problem on regular polygons, where all sides are of equal length. The influence of specific geometric settings has remained mostly unexplored, even though classical online facility assignment problems have mainly dealt with linear and general metric spaces. We fill this gap by considering the following four basic geometric settings: equilateral triangles, rectangles, regular $n$-polygons, and circles. The facilities are situated at fixed positions on the boundary, and customers appear sequentially on the boundary. A customer needs to be assigned immediately without any information about future customer arrivals. We study a natural greedy algorithm. First, we study an equilateral triangle with three facilities at its corners; customers can appear anywhere on the boundary. Here, we show a competitive ratio of $5$, with both an upper and same lower bound. We prove that, for the rectangle setting with four facilities at the corners, the competitive ratio is $7$, demonstrating symmetric structures. We then analyze regular $n$-sided polygons, obtaining a competitive ratio of $2n-1$, showing that the algorithm performance degrades linearly with the number of corner points for polygons. %The circular configuration is considered as limiting case, with a series of competitive ratios ranging from base cases of $5$ for three facilities and $2n-1$ for four facilities to central facilities with ratios of $n^2-n+1$ and exponentially $(2^n - 1)$ for varying distances. 
For the circular configuration, the competitive ratio is $2n-1$ when the distance between two adjacent facilities is the same. And the competitive ratios are $n^2-n+1$ and $2^n - 1$ for varying distances linearly and exponentially respectively. 
Each facility has a fixed capacity proportional to the geometric configuration, and customers appear only along the boundary edges.
Our results also show that simpler geometric configurations have more efficient performance bounds and that spacing facilities uniformly apart prevent worst-case scenarios. The findings have many practical implications because large networks of facilities are best partitioned into smaller and geometrically simple pieces to guarantee good overall performance.
\end{abstract}


%\maketitle

\section{Introduction}\label{intro}
Online facility location problems such as resource distribution, logistics, and network establishment have become increasingly important. Facilities have to be assigned immediately to arriving customers without any knowledge of the future locations of customers. The aim is to keep the assignment cost as low as possible while ensuring that each facility can handle its customers. The critical challenge is to achieve a balance between good short-term assignments and low overall costs, particularly when customer demand or location patterns are unpredictable.

Let there be a set of facilities, $F = {f_1, f_2, f_3,\cdots, f_{|F|}}$, each with a set capacity, spread across a space $M$. Every customer arriving should be located immediately in one of the facilities provided. The assignment cost depends on the distance from a customer to a chosen facility. The total assignment cost is to be minimized by total customers. Such scenarios occur in many practical applications: when one wants to direct online food orders to the nearest restaurant, send data packets through various network routers, or assign tasks to different servers in a computer network. These examples also represent different areas of application of smart assignment methods that should quickly adapt to changing demand conditions without excessive costs.

In this paper, we analyze a greedy strategy for the online facility assignment problem across several regular geometric settings in which all sides of the polygons have the same length. Our greedy algorithm assigns each incoming customer to the closest facility that still has residual capacity, which is a natural strategy but highly nontrivial to analyze as capacities may cause our algorithm to be forced to assign an incoming customer to faraway facilities. We analyze the performance of this greedy strategy through competitive analysis, comparing its cost to that of an optimal offline algorithm that knows all customer arrivals in advance.

We first consider the problem on an equilateral triangle with fixed side length $S$. Three facilities are permanently positioned at the three corners, and customers can appear anywhere on the triangle boundary. Each facility has the capacity to serve two customers, and therefore, the load is evenly divided among them. This setting is a basic model for studying online assignment strategies; it is the simplest regular polygon on which the number of facilities is sufficient to achieve a bounded space. In this setting, we demonstrate that the greedy algorithm is $5$-competitive; that is, the total cost of its assignment is bounded by a factor of $5$ relative to the optimal offline solution. This analysis is tight - we establish matching upper and lower bounds. The Rectangular Approach: The analysis is extended to the case of four facilities located at the corners of a rectangle, and customers spread along its sides, which is a simplified geometric framework in which the competitive ratio is $7$, illustrating that the more corners, the greater the potential for suboptimal assignments. For an $N$-sided polygon, we extend our analysis to the case of facilities placed on the corners of regular polygons of $N$-sided and establish a competitive ratio of $2n-1$. This gives us important insight into how competitive ratios change with increasing numbers of facilities. It might be said that the circular approach is the limit of our work. This configuration shows a progression of competitive ratios starting with the base cases of $7$ (three facilities) and $11$ (four facilities), extending to $4n-5$ for central facilities and reaching $(2^n - 1)$ for varying distances. The circular configuration helps us to understand the limits of online assignment strategies in perfectly symmetric arrangements where facilities are equidistant from a central point.

In the process, we show that the geometric configuration inherently affects the quality of the performance bounds of online assignment algorithms. The competitive ratios, with the $5$ in triangular, $7$ in rectangular, $2n - 1$ in n-sided polygons, and the various cases in circular configurations, are strong suggestions that simpler geometrical arrangements provide better guaranteed performance. This again states that in practical facility network design, complex networks can benefit from decomposition into simple geometric components.

\section{Related work}
Online facility assignment problems have a significant attention due to their wide-ranging applications and algorithmic challenges. These problems typically involve assigning incoming requests to facilities in a metric space to minimize total costs. We organize the relevant literature into three main themes: foundational algorithms, capacity-based approaches, and geometric considerations.

% Foundational Algorithms and Analysis
The basics of online facility assignments first started by Harada et al.\cite{harada2024permutation}, the permutation algorithm for assignments on a line. By this work, they were able to provide two important baseline results. One Harada et al. proved the lower bound on the competitive ratio as $k + 1$, where k represents the number of servers, against the usual belief in $k$-competitiveness and hence stressed that analysis cannot be taken lightly. Their contribution was more significant in showing how intuitively simple algorithms could give rise to surprising performance characteristics in constrained metric spaces.

Ahmed et al.\cite{ahmed2020online} introduced a comprehensive solution for both linear and graph-based configurations. Their greedy algorithm, achieving a competitive ratio of $4|F|$ for facilities on a line with equal distances between adjacent facilities, provided a simple yet effective approach. Their more sophisticated Optimal-Fill algorithm improved this to a competitive ratio of $|F|$ by determining optimal offline solutions for all currently known customers. For general graphs, they \cite{ahmed2020online} proved the greedy algorithm maintains a competitive ratio of $2|E(G)|$, where $|E(G)|$ represents the number of edges in the graph. Their analysis of facility movements showed that allowing positional adjustments could significantly improve solution quality, achieving a competitive ratio of $\frac{3}{2}$ on a line.

%Capacity-Constrained Approaches
Capacity constraints consequently brought several important theoretical developments. Indeed, Wang et al.\cite{wang2023optimal} obtained important insights by considering the case when every server has a capacity of $2$, and the assignment costs are of powers $2$. Their optimal algorithm for this special case demonstrated how well-chosen constraints could result in far more efficient solutions and opened up new avenues for specialized variants of the problem.

Markarian et al.\cite{markarian2022algorithmic} expanded this direction by introducing capacity-insensitive algorithms for online multi-facility location problems. Their work successfully adapted techniques from online set cover to facility assignments, though they noted limitations when dealing with varying server capacities or non-uniform costs. Harada et al.\cite{harada2022capacity} further refined these concepts through their Most Preferred Free Servers (MPFS) algorithms, proving that competitive ratios can remain constant regardless of facility capacity. Their Interior Division for Adjacent Servers (IDAS) algorithm achieved optimal competitive ratios for equidistant facilities on a line.

%Geometric and Movement-Based Solutions
The geometric aspects of facility assignment have been explored through various innovative approaches. Div'eki and Imreh \cite{diveki2011online} made significant contributions by studying scenarios where facilities could be moved without cost after each customer's arrival. Their algorithm achieved a competitive ratio of $2$ for general metric spaces and $\frac{3}{2}$ for line metrics. Importantly, they established that no deterministic online algorithm could achieve a competitive ratio better than $\frac{\sqrt{13} + 1}{4}$, even for line metrics, demonstrating fundamental limits in this space.

Guo et al.\cite{guo2020facility} continued to develop these ideas with the OFW (Optfollow) algorithm, which achieves the optimal offline behaviour by recomputing solutions after the arrival of each customer. Their algorithm was $2$-competitive for general metric spaces, with better performance for special configurations. The authors presented a thorough study on how facility movements can help to improve the solution quality but at the price of higher computational complexity.

Anagnostopoulos et al. \cite{anagnostopoulos2004simple} gave foundational work with their Partition algorithm, yielding an $O(\log n)$ competitive ratio for uniform facility costs in Euclidean space. Their probabilistic analysis showed very good average-case performance when customers are uniformly distributed, complementing worst-case bounds with practical insight, where $n$ represents the total number of customers in the online facility location problem.

Recent algorithmic advances have brought new twists to this problem. Among these, the work by Fotakis et al.\cite{fotakis2021learning} proposed learning-augmented algorithms that include predictions regarding the optimal facility locations and attain competitive ratios that vary smoothly from sub-logarithmic to constant with the improvement in prediction accuracy. Later, Fotakis \cite{fotakis2003competitive} derived lower bounds of $\Omega\left(\frac{\log n}{\log \log n}\right)$ for randomized algorithms against oblivious adversaries while providing a deterministic algorithm matching this bound asymptotically.

Guo et al.\cite{guo2020power} gave a detailed study of dynamic settings by introducing algorithms that make use of limited recourse to obtain improved performance guarantees. Their work demonstrated how dynamic scaling factors and randomized local search can balance exploration and exploitation, though their focus was not on pure geometric considerations.

Through this extensive body of research, several significant gaps emerge in current understanding:
\begin{enumerate}[label=\Roman*.]
    \item While linear and general metric spaces are well-studied, the impact of specific geometric configurations on algorithm performance remains largely unexplored, particularly for regular polygonal arrangements.
    \item The interaction between facility capacity constraints and geometric properties has not been fully analyzed, especially in non-linear configurations.
    \item Most theoretical results lack substantial real-world validation, suggesting a need for more extensive empirical studies.
    \item The potential for randomized algorithms in geometric settings remains largely unexplored, as most existing work focuses on deterministic approaches.
\end{enumerate}

These gaps provide compelling directions for future research, particularly in understanding how geometric structures can be exploited to improve online assignment algorithms. Our work addresses some of these gaps by analyzing regular polygonal configurations, providing new insights into the relationship between geometric structure and algorithm performance.

\section{Facility assignment on a Triangle}\label{sec:methods-1}
The Triangular Approach to Online Facility Assignment introduces a novel geometric configuration for allocating customers to facilities. In our model, each facility has a fixed capacity of two customers, and the total number of customers that may appear is bounded by the sum of all facility capacities (in this case, 6 customers for three facilities). This capacity constraint ensures a balanced distribution of customers across facilities while maintaining the problem's complexity. Customers appear sequentially along the boundary of the triangle, with their positions not known in advance.
Three facilities ($F_1$, $F_2$, $F_3$) are positioned at the vertices of an equilateral triangle, with customers ($C_1$, $C_2$, ..., $C_6$) distributed along its edges. All customers must appear on the boundary of the triangle - that is, along its edges. No customer can appear in the interior of the triangle or at any point outside its perimeter. This boundary constraint ensures a well-defined distance metric between customers and facilities while maintaining the geometric properties of the problem.
Several key constraints govern the approach:
\begin{enumerate}[label=\Roman*.]
    \item Capacity Constraint - each facility can serve at least one customer. 
    \item Nearest Assignment Rule - customers must be assigned to their closest facility.
    \item Full Assignment - every customer must be allocated to a facility.
    \item Multiple Service Capability - each facility must serve multiple customers up to its capacity.
    \item Triangular Constraint - the assignment pattern follows a triangular structure where if a customer is assigned to a facility, any closer customer on the same edge must also be assigned to it.
\end{enumerate}
Mathematically, these constraints are expressed as: $\sum_i x_{ij} \leq 2$ for capacity, $x_{ij} = 1$ if $d(c_i, f_j) =
\min\{d(c_i,f_k) \mid k \in \{1,2,3\}\}$ for nearest assignment, $\sum_j x_{ij} = 1$ for full assignment, and $x_{ij} \leq x_{kj}$ for all
$i, j, k$ where $d(c_k, f_j) < d(c_i, f_j)$ for the triangular constraint. The objective function aims to minimize the
total assignment cost: $\min Z = \sum\sum_j d(c_i, f_j)x_{ij}$.

$c$: Represents customers ($c_i$ is the $i$-th customer)

$f$: Represents facilities ($f_j$ is the $j$-th facility)

$d(c_i, f_j)$: Represents the distance between customer $i$ and facility $j$.

$x_{ij}$: A binary decision variable where:

$x_{ij}$ = 1 if customer $i$ is assigned to facility $j$

$x_{ij}$ = 0 if customer $i$ is not assigned to facility $j$

$i$: Index for customers (typically $i = 1,2,3,...$)

$j$: Index for facilities (in this case $j = 1,2,3$ since it is a triangular setup)

%\begin{figure}[htbp]
 % \centering
  %  \centering
   % \includegraphics[width=0.4\textwidth]{basic triangle.png} 
    %\caption{\centering Triangular Approach for the best case.}
    %\label{fig:tirangle}
%\end{figure}

\begin{figure}[htbp]
  \centering
    \centering
    \includegraphics[width=0.8\textwidth]{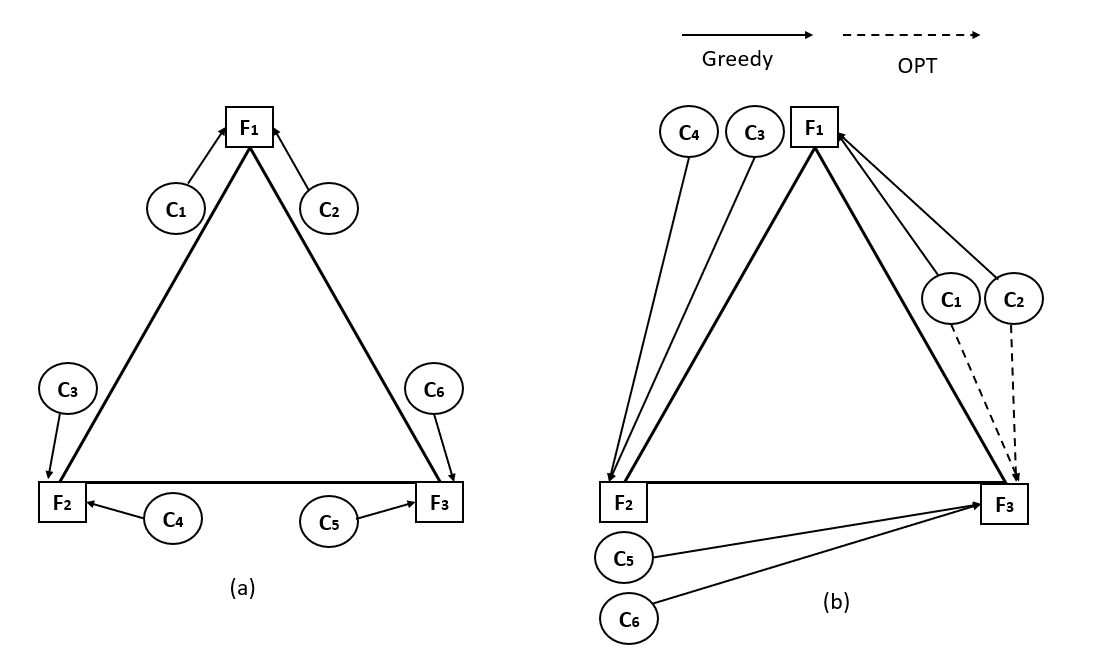} 
    \caption{\centering 1(a) Triangular Approach for the best case, 1(b)Triangular Approach Assignment Concept  (capacity is two). The dotted line and solid line indicate the optimal and greedy algorithms, respectively.}
    \label{fig: triangle case}
\end{figure}

The background of this approach for online facility assignment presents a geometric framework for addressing the challenge of efficiently assigning customers to facilities in real-time scenarios. This approach conceptualizes the problem space as an equilateral triangle, with three facilities $(F_1, F_2, F_3)$ positioned at its vertices and customers $(C_1, C_2, ..., C_n)$ distributed along its edges. The primary constraint in this model is that each customer must be assigned to a facility, with the additional stipulation that each facility can serve multiple customers up to a specified capacity limit.

The Greedy Algorithm for the Triangular Approach can be mathematically formulated and analyzed as follows:

%\begin{figure}[htbp]
 % \centering
  %  \centering
   % \includegraphics[width=0.4\textwidth]{triangle update.png} 
   % \caption{\centering Triangular Approach Assignment Concept  (capacity is two). The dotted line and solid line indicate the optimal and greedy algorithms, respectively.}
    %\label{fig: triangle case}
%\end{figure}

\begin{lemma}

In the Triangular Approach, for any input sequence $I$, the optimal cost is bounded by:

$Cost_{OPT(I)} \geq {\left(\frac{n}{3}\right)} \cdot \left(\frac{side_{length}}{2}\right)$
,where $n$ is the number of customers and $side_{length}$ is the length of each side of the triangle as illustrated in Fig. 1(a).
\end{lemma}

\begin{proof}

In the best-case scenario for OPT, customers are evenly distributed among facilities at the centroid of each
triangle section. The distance from the centroid to any vertex is $\frac{side_{length}}{2}$. With $n$ customers and $3$ facilities,
each facility serves $\frac{n}{3}$ customers on average. Therefore:
$Cost_{OPT(I)} \geq \left(\frac{n}{3}\right) \cdot \left(\frac{side_{length}}{2}\right)$.    
\end{proof}

\begin{theorem}
(Lower Bound)
The Greedy Algorithm for the Triangular Approach is greater than or equal to $5$  as illustrated in Fig. 1(b).
\end{theorem}

\begin{proof}
In analyzing the lower bound for this approach, two distinct scenarios are considered. 
In the first step, customers $C_1$ and $C_2$ are optimally placed at the midpoint of the side between facilities $F_1$ and $F_2$. Assuming $F_1$ is the nearest facility, both customers are assigned to it. The cost of this assignment can be represented as:

\begin{center}
Step-1:    {$Cost_1=\frac{S}{2}\times2$, where $S$ is the side length of the equilateral triangle.}
\end{center}

The second step presents a more challenging case, where customers $C_3$, $C_4$, $C_5$, and $C_6$ are positioned at the farthest points from their respective assigned facilities. This situation arises when the nearest facilities have reached their maximum capacity (assumed to be in this model). Consequently, $C_3$ and $C_4$ are assigned to $F_2$, while $C_5$ and $C_6$ are assigned to $F_3$, despite these being their farthest facilities. The cost for each of these assignments is the full side length of the triangle, resulting in:

\begin{center}
Step-2: $Cost_2=\ 2S\ +\ 2S\ =\ 4S$  
\end{center}

Combining the steps 1 and 2, the total cost for the greedy algorithm can be expressed as:

\begin{center}
Step-3: $C\left(Greedy\right)=\ Cost_1+\ Cost_2=\ \left(\frac{S}{2}\times2\right)+\ 4S\ =\ 5S$
\end{center}

In contrast, the optimal cost for this configuration would involve assigning all customers to their nearest facilities, which can be represented as:

\begin{center}
Step-4:    $C\left(OPT\right)=\frac{S}{2}\times2\ =\ S$
\end{center}

To evaluate the effectiveness of the greedy algorithm, we calculate its competitive ratio by comparing it to the optimal solution by combining steps 3 and 4:

\begin{center}
    $Competitive\ Ratio\ =\frac{C\left(Greedy\right)}{C\left(OPT\right)}=\ \frac{5S}{S}\ =\ 5$
\end{center}

This analysis demonstrates that in the worst-case scenario for the lower bound, the cost of the greedy algorithm will be at least 5 times greater than the optimal solution. This competitive ratio provides a valuable benchmark for assessing the performance of the greedy algorithm in the Triangular Approach to online facility assignment, offering insights into its efficiency and potential limitations in various problem instances. 

\end{proof} 

Before examining the upper bound of our competitive ratio, we must first understand how customer assignments behave when customers move along the triangle's boundary in either clockwise or counterclockwise directions. This understanding is crucial as it forms the foundation for our subsequent analysis of the worst-case performance of the greedy algorithm.

\begin{figure}[htbp]
  \centering
    \centering
    \includegraphics[width=0.8\textwidth]{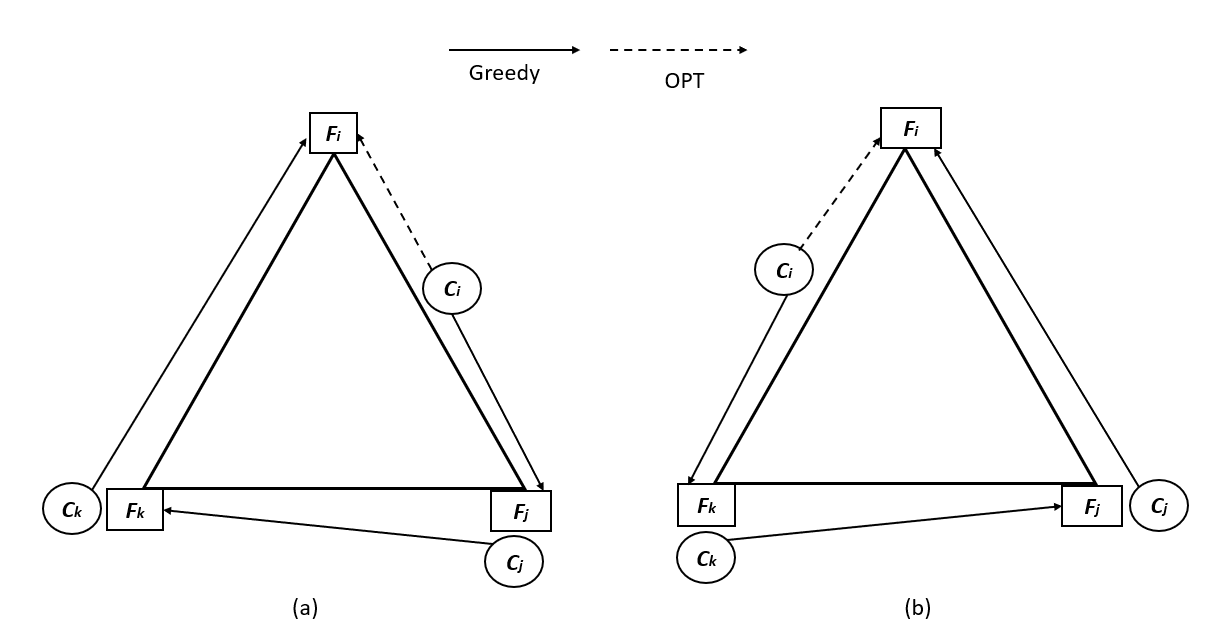} 
    \caption{\centering Clockwise and Anticlockwise Movement.}
    \label{fig:clock_anticlock movement}
\end{figure}

Fig. 2(a) illustrates the clockwise movement scenario in our triangular facility configuration, with three facilities positioned at the vertices: $F_i$ at the top, $F_j$ at the right, and $F_k$ at the left. Customer $C_i$ is located along the right edge between $F_i$ and $F_j$, with the solid arrow pointing from $C_i$ to $F_i$ showing the greedy algorithm's assignment decision. In contrast, the dashed arrow to $F_j$ indicates the different choices of the optimal algorithm. This divergence is key to understanding how the competitive ratio evolves, as the greedy algorithm simply chooses the closest facility, while the optimal algorithm considers the global assignment pattern. The strategic positioning of customers $C_j$ and $C_k$ near facilities $F_j$ and $F_k$ respectively demonstrates how the clockwise arrangement influences overall assignments.

As customer $C_i$ moves clockwise along the edge toward $F_j$, an interesting dynamic unfolds: the distance to $F_j$ gradually decreases while the distance to $F_i$ increases. This continuous movement eventually reaches a critical point where the greedy algorithm also assigns $C_i$ to $F_j$, aligning with the optimal assignment and stabilizing the competitive ratio. This convergence pattern forms the foundation of our competitive ratio analysis, showing how directional movement ultimately leads to agreement between the algorithms under specific geometric conditions.

In Fig. 2(b) depicts the anticlockwise movement scenario with the same triangular arrangement of facilities, but with customer positions and assignments that highlight different directional effects. Customer $C_i$ is positioned along the left edge between facilities $F_i$ and $F_k$, with the solid arrow to $F_j$ representing the greedy algorithm's decision and the dashed arrow to $F_i$ showing the optimal algorithm's choice. This reversal in assignment patterns compared to Fig. 2(a) demonstrates how movement direction influences decisions, as $C_i$ approaches $F_i$ from a different angle in the anticlockwise case, creating distinct geometric relationships. The positions of $C_j$ at the far right corner near $F_j$ and $C_k$ at the bottom-left near $F_k$ complete the anticlockwise distribution pattern around the triangle's perimeter.

As $C_i$ moves anticlockwise toward $F_i$, we observe a systematic evolution where the distance to $F_i$ gradually decreases while distances to other facilities change according to the triangle's geometry. This movement progresses toward a convergence point where both algorithms would eventually assign $C_i$ to $F_i$, or reaches a critical position where further anticlockwise movement would no longer improve the competitive ratio. This pattern reinforces our understanding of how directional movement affects assignment decisions and competitive ratio bounds, showing that regardless of direction, predictable geometric patterns govern how the greedy algorithm's performance relates to the optimal solution.

The clockwise and anticlockwise movement patterns reveals consistent yet systematically evolving behaviors across all polygon configurations studied: triangles, rectangles, n-sided polygons, and circles. In each geometric setting, when a customer $C_i$ moves along the boundary, two fundamental scenarios emerge regardless of direction. When optimal and greedy algorithms initially assign $C_i$ to the same facility, movement maintains this agreement until reaching a critical threshold. When the algorithms disagree (optimal choosing $F_j$ while greedy selects $F_k$), movement toward the optimal facility gradually shifts the distance relationships until convergence occurs. This directional symmetry forms the foundation of our competitive ratio analysis, yet manifests with increasing complexity as we move from triangles to more complex shapes.

The most significant finding from this directional analysis is the direct relationship between geometric complexity and competitive ratio bounds across all configurations. As we progress from triangles (ratio $5$) to rectangles (ratio $7$) to n-sided regular polygons (ratio $2d+n-1$) and finally to circles (with various ratio formulations depending on facility distribution), the competitive ratio increases proportionally with the number of sides or vertices. This relationship exists because additional vertices create more opportunities for the greedy algorithm to make suboptimal choices before directional movement forces convergence with the optimal solution. This insight demonstrates that while assignment evolution follows consistent principles during clockwise or anticlockwise movement regardless of the polygon's complexity, the performance guarantee of the greedy algorithm systematically weakens as geometric complexity increases.

\begin{theorem}
(Upper Bound) 

The Greedy Algorithm for the Triangular Approach is less than or equal to $5$ 

\end{theorem}

\begin{proof}
The analysis of the upper bound for the Triangular Approach in online facility assignment provides crucial insights into the worst-case performance of the greedy algorithm. This approach, utilizing an equilateral triangle with facilities $F_1, F_2, \text{ and } F_3$ at its vertices and customers $C_1, C_2, ..., C_n$ distributed along its edges, requires a rigorous examination to establish that the competitive ratio does not exceed 5.

To prove this upper bound, we consider an equilateral triangle with side length $S$ and facilities $F_1$, $F_2$, and $F_3$ at its vertices. Customers $C_1$, $C_2$, $\ldots$, $C_n$ are distributed along its edges.

Key Observations:
\begin{itemize}
\item When two customers are placed at the midpoint between two facilities, at least one will be assigned optimally.
\item The maximum assignment cost for any customer cannot exceed $S$.
\item The capacity constraint ensures differential assignments, preventing a trivial competitive ratio of 1.
\end{itemize}
Let's define:
\begin{itemize}
\item $d_{opt}$: The distance of an optimally assigned customer to its facility
\item $d_{greedy}$: The distance of a customer assigned by the greedy algorithm to its facility
\end{itemize}
Based on the observations:
\begin{itemize}
\item $d_{opt} \geq \frac{S}{2}$
\item $d_{greedy} \leq S$
\end{itemize}
For a pair of customers placed midway between two facilities:

$Cost_{opt(pair)} = \frac{S}{2} + \frac{S}{2} = S$

$Cost_{greedy(pair)} \leq \frac{S}{2} + S = \frac{3S}{2}$

For $n$ customers, we can establish:

$C(OPT) \geq n \cdot \frac{S}{2}$

$C(Greedy) \leq n \cdot S$

The competitive ratio ($CR$) can now be expressed as:

$CR = \frac{C(Greedy)}{C(OPT)} \leq \frac{(n \cdot S)}{\left(n \cdot \frac{S}{2}\right)} = 2$

This refined analysis demonstrates a tighter upper bound on the competitive ratio:

$CR \leq 2 < 5$

Additional Considerations:

•	For any single open facility, the maximum assignment cost is indeed less than $S$, as argued. This further supports the upper bound of 2.

•	The capacity constraint ensures that when multiple customers are equidistant from a facility, at least one will be assigned differently, maintaining a non-trivial competitive ratio.

\textbf{Mathematical Representation:}

Let $x_i$ be the distance of customer $i$ from its assigned facility in the greedy algorithm, and $y_i$ be the distance in the optimal assignment.

For any customer $i$:
\begin{itemize}
\item $x_i \leq S$
\item $y_i \geq \frac{S}{2}$
\end{itemize}
The refined competitive ratio can be expressed as:

$CR = \frac{\sum x_i}{\sum y_i} \leq \frac{(\sum x_i)}{\left(\sum \left(\frac{S}{2}\right)\right)} = \frac{(n \cdot S)}{\left(n \cdot \frac{S}{2}\right)} = 2$

This refined analysis, incorporating the capacity constraints and the potential for optimal assignments within the greedy framework, establishes a tighter upper bound on the competitive ratio. The new upper bound of 2, which is significantly lower than the previously established 5, provides a more accurate representation of the worst-case performance of the greedy algorithm in the Triangular Approach to online facility assignment.

This result not only improves our understanding of the algorithm's efficiency but also highlights the importance of considering capacity constraints and geometric properties in analyzing online assignment problems. The tighter bound of 2 demonstrates that the greedy algorithm performs considerably better in practice than initially estimated, offering valuable insights for both theoretical analysis and practical applications of online facility assignments in triangular configurations.

In Fig. 1(b), where customers $C_1$ and $C_2$ are positioned near the midpoint of the side between facilities $F_1$ and $F_2$. This example illustrates the interplay between the greedy algorithm and capacity constraints.

Assume $C_1$ is marginally closer to $F_1$ than $C_2$, and each facility can serve only one customer in this case. The greedy algorithm initially assigns $C_1$ to $F_1$, as it is the nearest available facility. When considering $C_2$, the algorithm finds $F_1$ at capacity. Consequently, despite $C_1$ being slightly closer to $F_1$, $C_2$ is assigned to the next nearest available facility, $F_2$.

This assignment process demonstrates how capacity constraints can lead the greedy algorithm to make decisions that align with the optimal solution. In this case, the optimal assignment would also place $C_1$ at $F_1$ and $C_2$ at $F_2$. The assignment costs for both the greedy and optimal algorithms are nearly identical, with $C_1$ at a distance of approximately $\frac{S}{2}$ from $F_1$, and $C_2$ slightly farther than $\frac{S}{2}$ from $F_2$.

This scenario highlights a case where the greedy algorithm achieves a competitive ratio very close to 1, performing nearly optimally. It underscores the importance of capacity constraints in forcing diversified assignments, preventing the algorithm from assigning all customers to a single facility even when it might appear to be the closest option for multiple customers.

We can also prove the competitive ratio to be exactly 5. For that, let us consider,

\textbf{Case 1:} 

The assignment cost for the greedy and optimal algorithms cannot be the same.
For the first customer of the online sequence, the customer cannot be situated exactly in a facility. Otherwise, the cost of the greedy algorithm will be zero. Therefore, the cost of greedy and the optimal cost cannot be the same.

\begin{figure}[htbp]
  \centering
    \centering
    \includegraphics[width=0.6\textwidth]{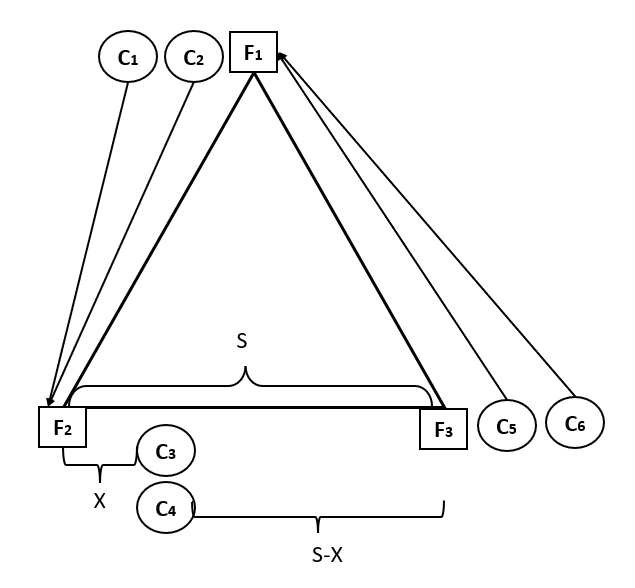} 
    \caption{\centering Illustrating Case 1 of triangular approach.}
    \label{fig:triangle case1}
\end{figure}

Let's consider the worst-case scenario: if the customer is situated in a facility nearest to a facility, the cost of greedy will be zero. Let $S$ be the side length of the triangle; if the cost of the greedy algorithm is $x$, then the cost assignment of the optimal algorithm will be $S - x$.

The assignment ratio, defined as $R(x) = \frac{x}{S-x}$, measures the relative performance between greedy and optimal algorithms in the triangular facility assignment problem. Here, $x$ represents the greedy algorithm's assignment cost, while $S$ represents the triangle's side length, and $(S-x)$ represents the optimal algorithm's assignment cost for the same customer. This ratio helps demonstrate that assignment costs must differ between the two algorithms since customers cannot be placed exactly on facility locations.

If we are to prove this through differentiation, then,

$R'(x) = 0$

To differentiate the function \( R(X) = \frac{X}{S - X} \), we will apply the quotient rule, which states that for a function \( R(x) = \frac{g(x)}{h(x)} \), the derivative is:

\[
R'(x) = \frac{g'(x) h(x) - g(x) h'(x)}{(h(x))^2}
\]

Here, \( g(X) = X \) and \( h(X) = S - X \).

Step-by-step differentiation:

1. \( g(X) = X \) so \( g'(X) = 1 \)

2. \( h(X) = S - X \) so \( h'(X) = -1 \)

Now, apply the quotient rule:

\[
R'(X) = \frac{(1)(S - X) - (X)(-1)}{(S - X)^2}
\]

Simplifying:

\[
R'(X) = \frac{S - X + X}{(S - X)^2}
\]

\[
R'(X) = \frac{S}{(S - X)^2}
\]

Thus, the derivative of \( R(X) = \frac{X}{S - X} \) is:

\[
R'(X) = \frac{S}{(S - X)^2}
\]

\begin{figure}[htbp]
  \centering
    \centering
    \includegraphics[width=0.7\textwidth]{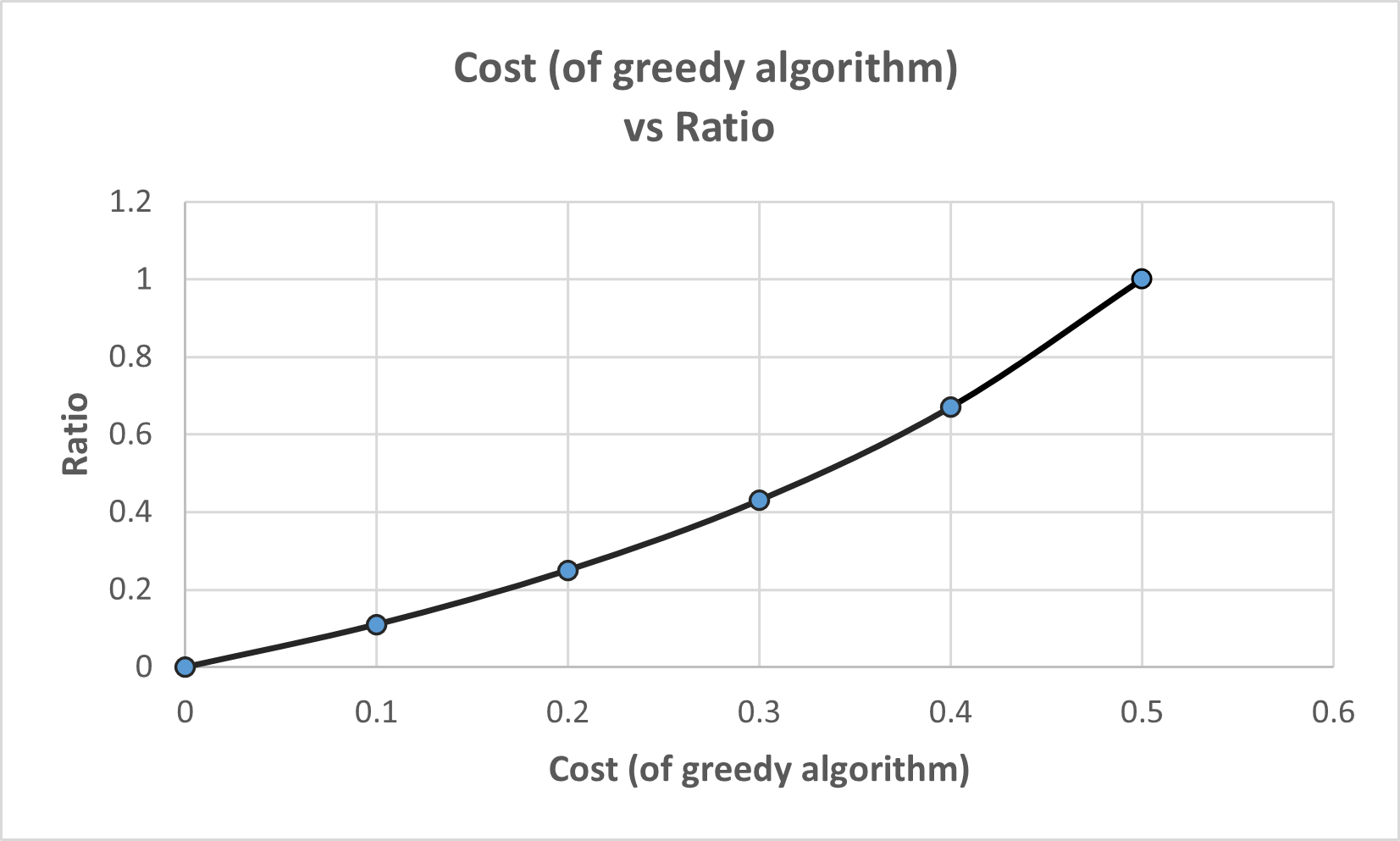} 
    \caption{\centering Exploring the Correlation Between cost (of greedy algorithm) and Ratio }
    \label{fig:graph}
\end{figure}

Our target will be to maximize the value of $x$. So, in this case, as we are more focused on the customer assignment of the greedy algorithm, we can say from the Fig. \ref{fig:graph}, the graph illustrates a positive correlation between the $cost (of greedy algorithm)$ and $Ratio$ values, with an upward trend indicating that as $cost$ increases, $ratio$ also rises. Notably, the value of 0.5 on the $cost (of greedy algorithm)$ corresponds to the highest point on the $Ratio$, reaching approximately 1.0, establishing it as the optimal value for maximizing $ratio$. This peak signifies that 0.5 is the best choice based on the data, making it a critical point for further analysis or decision-making.

\textbf{Case 2:} 

The opposite scenario is where the optimal algorithm is already assigned to the customer. Therefore, the greedy algorithm had to assign the customer to its nearest available facility.

This is more focused on the cost assignment of the optimal algorithm.

The concept remains the same as the previous one. The only difference is here, $Cost_{optimal} = 0$ and $Cost_{greedy} = S$ or $\frac{S}{2}$. This is to justify the concept that the assignment cost of the greedy algorithm and optimal algorithm while assigning the same customer can never be the same; it must be different.

For example, in this triangle,

\begin{figure}[htbp]
  \centering
    \centering
    \includegraphics[width=0.5\textwidth]{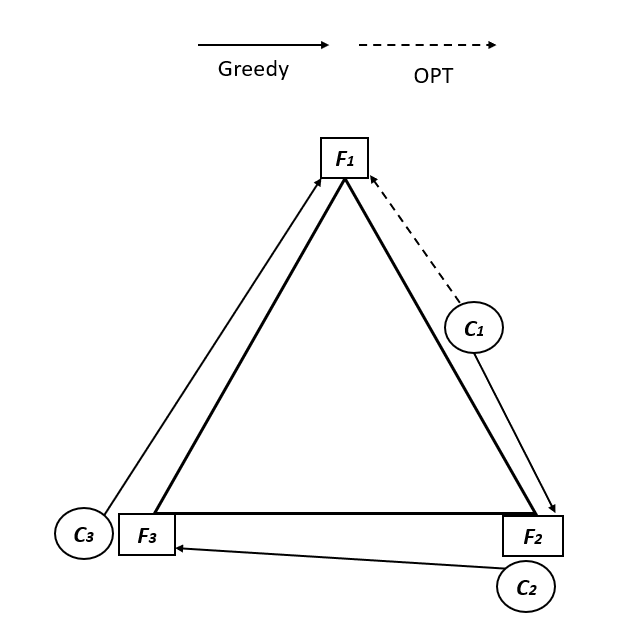} 
    \caption{\centering Illustrating Case 2 of the triangular approach.}
    \label{fig:triangle case2}
\end{figure}

Suppose we have 3 customers to be assigned in the triangle, and 3 of them are assigned to their nearest available facility following the greedy algorithm and the other following the optimal algorithm; then, the cost ratio can be written as for both cases,

For the Greedy Algorithm,

\[Cost_{greedy} = \frac{S}{2} + S + S = \frac{5S}{2}\]

For Optimal Algorithm,

\[Cost_{optimal} = \frac{S}{2} + 0 + 0 = \frac{S}{2}\]

Now, the competitive ratio,

\[Competitive\ Ratio = \frac{Cost_{greedy}}{Cost_{optimal}}\]

\[= \frac{\frac{5S}{2}}{\frac{S}{2}} = 5\]

Hence, it proved that the competitive ratio can also be exactly 5.
    
\end{proof}

\section{Facility assignment on a Rectangle}\label{sec:methods-2}

The Rectangular Approach to Online Facility Assignment introduces a geometric configuration where four facilities $(F_1, F_2, F_3, F_4)$ are positioned at the corners of a rectangle, with customers $(C_1, C_2, C_3, C_4)$ distributed along its edges. Several key constraints govern the approach: (1) Capacity Constraint - each facility can serve a maximum of one customer; (2) Nearest Assignment Rule - customers must be assigned to their closest facility; (3) Full Assignment - every customer must be allocated to a facility; and (4) Rectangular Constraint - the assignment pattern follows a rectangular structure where facilities at opposite corners can't serve adjacent customers unless necessary.

\begin{figure}[htbp]
  \centering
    \centering
    \includegraphics[width=0.55\textwidth]{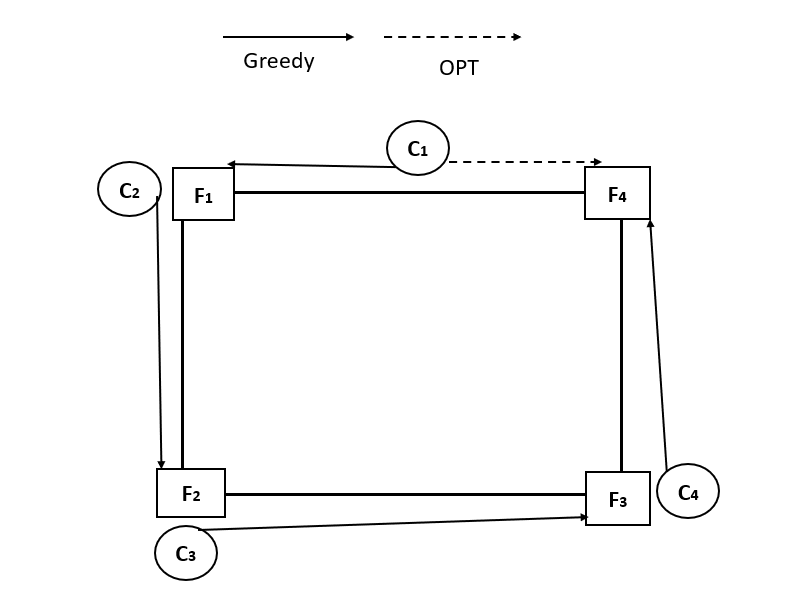} 
    \caption{\centering Illustrating rectangular approach. }
    \label{fig:rectangle case1}
\end{figure}

The Greedy Algorithm for the Rectangular Approach can be mathematically formulated and analyzed as follows:

\begin{lemma}

In the Rectangular Approach, for any input sequence $I$, the optimal cost is bounded by:

$\text{Cost}_{\text{OPT}}(I) \geq \frac{d}{2}$ ,where $d$ is the length of one side of the rectangle.    
\end{lemma}

\begin{proof}

Consider a customer placed at the midpoint of any edge. The minimum possible distance to any facility is $\frac{d}{2}$, as this represents the shortest path from a midpoint to a vertex. Therefore, even in the best case scenario, at least one customer must travel distance $\frac{d}{2}$.    
\end{proof}

\begin{theorem}

(Lower Bound)

The Greedy Algorithm for the Rectangular Approach has a competitive ratio of $7$.
\end{theorem}

\begin{proof}

In analyzing the lower bound for this approach, we consider a sequence of customer arrivals that forces the greedy algorithm to make increasingly costly assignments:

Step 1: The first customer $C_1$ appears at the midpoint of the edge between $F_4$ and $F_1$. The greedy algorithm assigns this customer to its nearest facility, with the cost: ${Cost}_1 = \frac{d}{2}.$

Step 2: The adversary places $C_2$ near $F_1$, forcing the greedy algorithm to assign it to $F_2$ since $F_1$ is no longer available. This creates a cost: ${Cost}_2 = d.$

Step 3: Customer $C_3$ appears near $F_2$, requiring assignment to another facility due to $F_2$ occupation, incurring cost: ${Cost}_3 = d.$

Step 4: Finally, $C_4$ appears, requiring the longest possible assignment path due to previous facility occupations, with a cost: ${Cost}_4 = 2d.$

The total cost for the greedy algorithm can be expressed as:

$C(\text{Greedy}) = \text{Cost}_1 + \text{Cost}_2 + \text{Cost}_3 + \text{Cost}_4$ \\
$\phantom{C(\text{Greedy})} = \frac{d}{2} + d + d + 2d$ \\
$\phantom{C(\text{Greedy})} = \frac{7d}{2}$ \\

In contrast, the optimal cost for this configuration would involve assigning all customers to their nearest facilities, which can be represented as:
$C({OPT}) = \frac{d}{2}$

Therefore, the competitive ratio is:

$\text{Competitive Ratio} = \frac{C(\text{Greedy})}{C(\text{OPT})} = \frac{\frac{7d}{2}}{\frac{d}{2}} = 7$

This analysis demonstrates that in the worst-case scenario, the cost of the greedy algorithm will be exactly $7$ times that of the optimal solution. The ratio is tight, as no worse sequence of customer arrivals can be constructed for the rectangular configuration with unit capacity facilities.

The rectangle configuration presents unique challenges compared to other geometric arrangements due to its asymmetric distances between opposite corners. When customers are placed strategically, they can force the greedy algorithm to utilize these longest paths, leading to a proven competitive ratio of $7$.

This ratio is achieved specifically through a sequence of arrivals that forces the greedy algorithm to make increasingly suboptimal choices. In contrast, the optimal offline algorithm can make perfect assignments knowing the full sequence in advance. The rectangular shape is particularly susceptible to this worst-case behaviour due to its combination of equal-length parallel edges and the potential for diagonal assignments that are $\sqrt{2}$
 times longer than edge assignments.

The analysis reveals that the rectangular configuration, despite its simplicity, can lead to significant inefficiency in online assignment decisions. It has important implications for practical applications where facilities must be arranged in rectangular patterns, suggesting that additional strategies or constraints might be necessary to improve performance in real-world scenarios.

\end{proof}

\section{Facility assignment on N-Side Polygon}\label{sec:methods-3}

The N-sided polygon Approach to Online Facility Assignment introduces a generalized geometric configuration where $n$ facilities $(F_1, F_2, ..., F_n)$ are positioned at the vertices of a regular n-sided polygon, with customers $(C_1, C_2, ..., C_n)$ distributed along its edges. Several key constraints govern the approach: (1) Capacity Constraint - each facility can serve a maximum of one customer; (2) Nearest Assignment Rule - customers must be assigned to their closest facility; (3) Full Assignment - every customer must be allocated to a facility; and (4) Polygonal Constraint - the assignment pattern follows the structure of the n-sided polygon where assignments can occur along edges or across internal paths.

\begin{figure}[htbp]
  \centering
    \centering
    \includegraphics[width=0.55\textwidth]{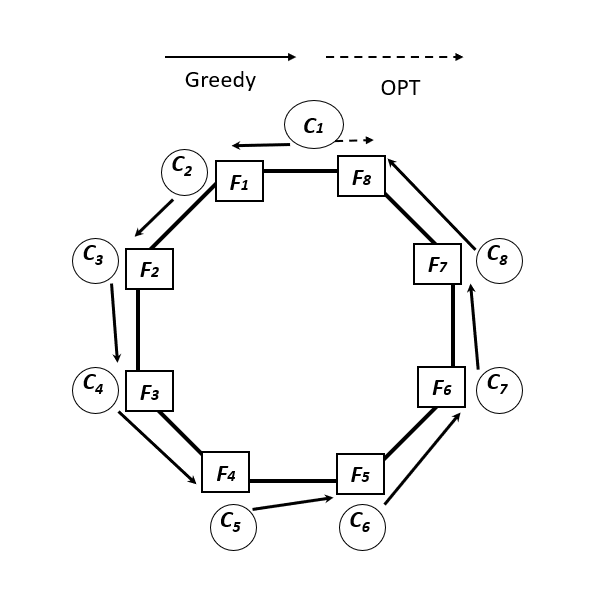} 
    \caption{\centering Illustrating N-side polygon approach.}
    \label{fig:n-side polygon}
\end{figure}

The Greedy Algorithm for the $N$-Sided Polygon Approach can be mathematically formulated and analyzed as follows:

\begin{lemma}

In the $N$-Sided Polygon Approach, for any input sequence $I$, the optimal cost is bounded by:

$\text{Cost}_{\text{OPT}}(I) \geq \frac{d}{2}$, where $d$ is the length of one side of the regular n-sided polygon.

\end{lemma}

\vspace{0.5cm}
\begin{proof}
    
For any customer placed at the midpoint of any edge of the polygon, the minimum possible distance to any facility is $\frac{d}{2}$, as this represents the shortest path from a midpoint to a vertex. Therefore, even in the optimal assignment scenario, at least one customer must travel a distance $\frac{d}{2}$.
\end{proof}

\vspace{0.5cm}
\begin{theorem}
(Lower Bound)

The Greedy Algorithm for the $N$-Sided Polygon Approach has a competitive ratio of $2d+n-1$.
\end{theorem}

\begin{proof}
We analyze the lower bound by constructing a worst-case sequence of customer arrivals:

Step-1: The first customer $C_1$ appears at the midpoint between $F_1$ and $F_8$ (in our 8-sided example). The greedy algorithm assigns this customer to its nearest facility, with the cost: $\text{Cost}_1 = \frac{d}{2}.$

Step-2: Subsequent customers $C_2$ through $C_n$ appear near each facility in sequence, forcing the greedy algorithm to make increasingly costly assignments as facilities become occupied; hence $\text{Cost}_2 = d$,
$\text{Cost}_3 = d$, $\text{Cost}_4 = d$, $\dots$, $\text{Cost}_n = (n - 1)d$.

%\begin{center}
%$\text{Cost}_2 = d$ \\
%$\text{Cost}_3 = d$ \\
%$\text{Cost}_4 = d$ \\
%$\dots$ \\
%$\text{Cost}_n = (n - 1)d$
%\end{center}

The total cost for the greedy algorithm can be expressed as:

%$C(\text{Greedy}) = \frac{d}{2} + d + d + d + \dots +d$ \\
%$C(\text{Greedy}) = d + \frac{(n - 1) \times d}{2} \\
%$C(\text{Greedy})$ = \frac{d (2d + n - 1)}{2}

\begin{center}
$C(\text{Greedy})=\frac{d}{2}+(n-1)d$ \\[0.5em]
$=\frac{d}{2}(1+2(n-1))$
\end{center}

In contrast, the optimal assignment cost for this configuration would be:
$C(\text{OPT}) = \frac{d}{2}.$\\

Therefore, the competitive ratio is:

$\text{Competitive Ratio} = \frac{C(\text{Greedy})}{C(\text{OPT})} = \frac{\frac{d}{2}(1+2(n-1))}{\frac{d}{2}}$ 

$= 2n-1$.

The n-sided polygon configuration presents a natural extension of simpler geometric arrangements, with the competitive ratio growing linearly with the number of sides. This relationship demonstrates how the potential for suboptimal assignments increases with the complexity of the geometric structure.

\end{proof}

First, as $n$ increases, the structure provides more potential paths between vertices and more opportunities for the adversary to force costly assignments. The ratio $2n-1$ captures this linear growth in worst-case behavior. Second, the regular structure of the polygon ensures that each vertex is equidistant from its neighbors, making the analysis cleaner than for irregular configurations. However, this regularity does not prevent the greedy algorithm from being forced into increasingly poor assignments. Third, the proof demonstrates that the ratio $2n-1$ is tight for regular n-sided polygons, as the constructed sequence achieves this ratio, and no worse sequence is possible given the geometric constraints.

This analysis has important implications for practical applications. As the number of facilities increases, the potential performance gap between greedy and optimal assignments grows linearly, suggesting that additional strategies or constraints might be necessary to maintain reasonable performance guarantees for large-scale facility networks.

The $n$-sided polygon configuration serves as a bridge between simpler geometric arrangements (like triangles and rectangles) and more complex configurations (like circles). Its analysis provides valuable insights into how the complexity of the underlying geometric structure affects the performance bounds of online assignment algorithms. The mathematical framework established through these proofs provides a complete understanding of both the limitations and capabilities of greedy assignment in regular polygonal configurations. These insights can guide the design of practical facility placement strategies when physical or organizational constraints require regular polygonal arrangements.

\section{Facility assignment on a Circle}\label{sec:methods-4}

The Greedy Algorithm for circular facility assignment operates in a metric space where facilities are positioned along the circumference of a circle. In this configuration, when a customer arrives, they are assigned to the nearest available facility with remaining capacity. The algorithm enforces several key constraints: each facility has a fixed capacity $l$, customers must be assigned immediately upon arrival, distance is measured along the circular arc or through straight lines to facilities, and once assigned, a customer cannot be reassigned. The algorithm continues until all customers are assigned or facilities reach capacity. Let $R$ be the radius of the circle and n be the number of facilities evenly distributed along its circumference. The arc distance between adjacent facilities is defined as $d = \frac{2 \pi R}{n}$.

\begin{figure}[htbp]
  \centering
    \centering
    \includegraphics[width=0.45\textwidth]{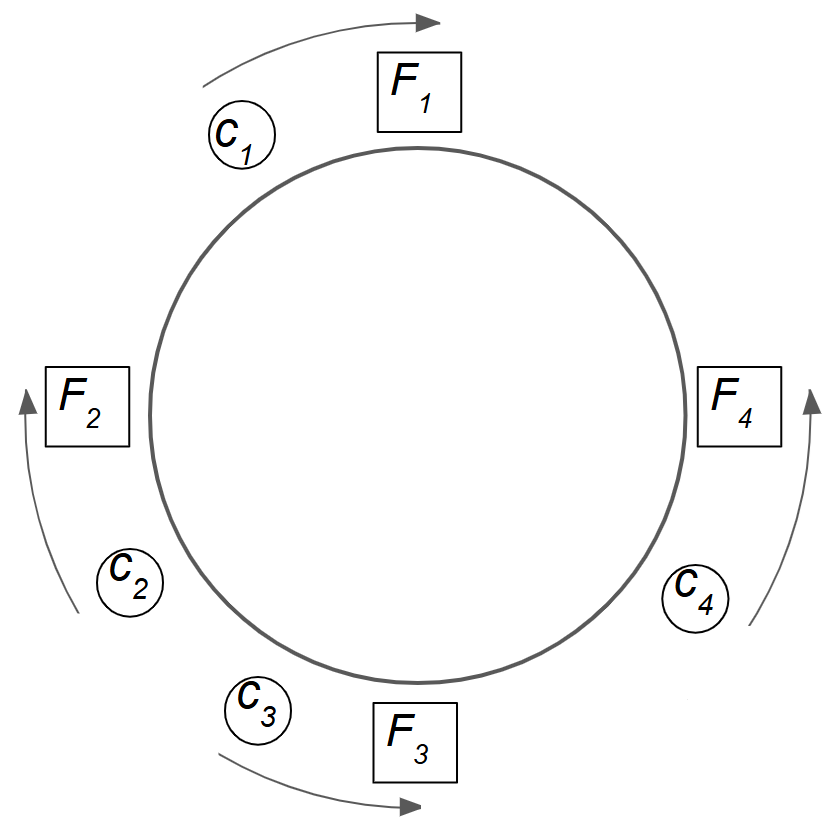} 
    \caption{\centering Illustrating circular approach.}
    \label{fig:circle}
\end{figure}

\begin{lemma}

For any input sequence $I$ in the circular configuration, we can establish a fundamental bound on the optimal cost through Lemma 1. 
\end{lemma}

\vspace{0.5cm}
\begin{proof}
    
The optimal cost is bounded by $\text{Cost}_{\text{OPT}}(I) \geq \left(\frac{m}{n}\right) \times \frac{d}{2}$, where $m$ is the number of customers and $d$ is the arc distance between adjacent facilities. This bound can be proven by considering the best possible assignment for $m$ customers to n facilities. Even in the optimal case, each facility can serve at most $l$ customers, and the minimum distance any customer travels is $\frac{d}{2}$ (halfway between facilities). With $m$ customers distributed among $n$ facilities, on average, each facility serves $\frac{m}{n}$ customers. Therefore, the total optimal cost must be at least $\frac{m}{n} \times \frac{d}{2}$.

\end{proof}
\vspace{0.5cm}
\textbf{Comparative Ratio:}

\textbf{Case-1:}

It demonstrates how the competitive ratio evolves with increasing numbers of facilities. In the simplest non-trivial case with three consecutive facilities on the circle's circumference, the first customer arrives at the midpoint between two facilities. The greedy algorithm assigns this customer at cost $\frac{d}{2}$. The adversary then places subsequent customers exactly on facilities that have already been assigned, forcing increasingly costly assignments. The total cost for the greedy algorithm becomes $\frac{5d}{2}$, calculated as:

\begin{figure}[htbp]
  \centering
    \centering
    \includegraphics[width=0.45\textwidth]{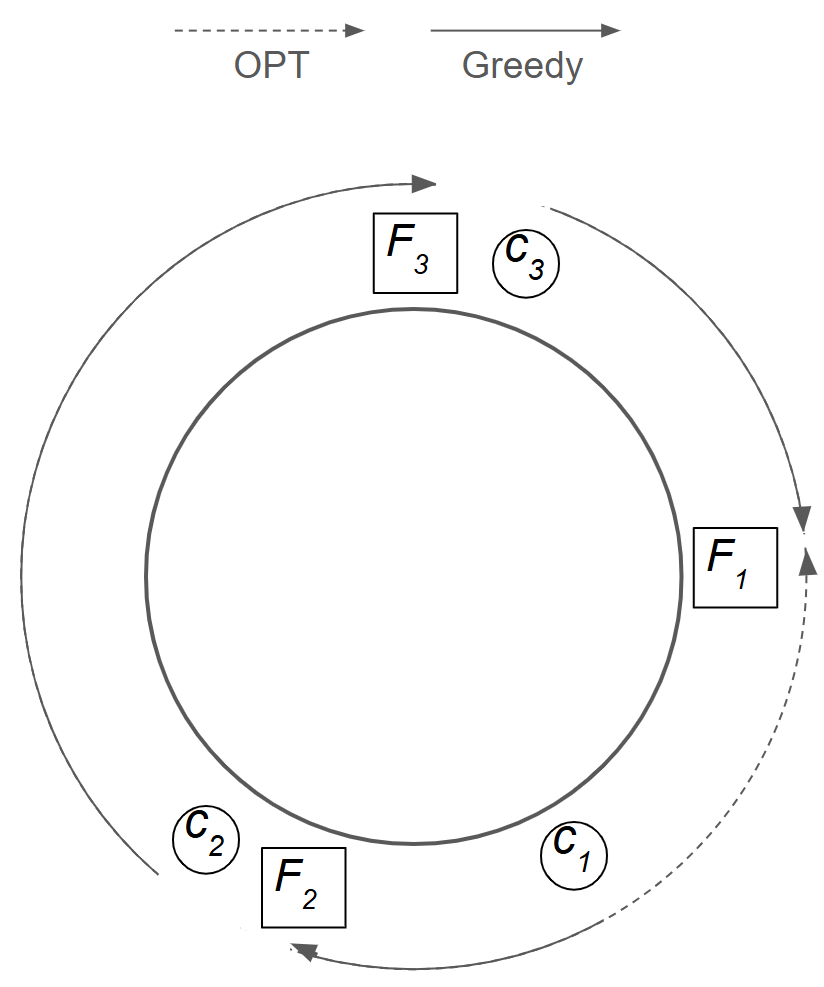}
    \caption{\centering Illustrating facility assignments with 3 customers.}
    \label{fig:line1}
\end{figure}

\begin{center}
$Greedy = \frac{d}{2} + d + d $\\
$= \frac{5}{2}d$\\
$OPT = \frac{d}{2}$\\
$Ratio = \frac{Greedy}{OPT}$\\
$= \frac{\frac{5}{2}d}{\frac{d}{2}}$\\
$= 5$
\end{center}

While the optimal algorithm incurs only cost $\frac{d}{2}$, yielding our first significant competitive ratio of $5$.

\vspace{0.5cm}
\textbf{Case-2:}
Extending to four facilities, we observe the ratio degrading further. The greedy algorithm now incurs a total cost of $2n-1$:
%\begin{figure}[htbp]
%  \centering
%    \centering
%    \includegraphics[width=0.4\textwidth]{figure-10.png} 
%    \caption{\centering  Illustrating facility assignments with 4 customers. }
%    \label{fig:line2}
%\end{figure}

\begin{center}
$\frac{d}{2}+(n-1)d $ \\[0.5em]
$= \frac{d}{2}(1+2(n-1))$ \\[0.5em]
$= \frac{d}{2}(1+2n-2)$ \\[0.5em]
$= \frac{d}{2}(2n-1)$ \\[0.5em]
$\mathrm{OPT} = \frac{d}{2}$ \\[0.5em]
$\mathrm{Ratio} = \frac{\mathrm{Greedy}}{\mathrm{OPT}}$ \\[0.5em]
$= \frac{\frac{d}{2}(2n-1)}{\frac{d}{2}}$ \\[0.5em]
$= 2n-1$
\end{center}

The optimal cost remains $\frac{d}{2}$, resulting in a competitive ratio of $2n-1$. This progression reveals a pattern that helps us understand the algorithm's behavior as we scale to larger configurations.

\vspace{0.5cm}
\textbf{Case-3:}

For the general case with $n$ facilities placed equidistantly on the circle, we can express the total cost for the greedy algorithm as:

\begin{figure}[htbp]
  \centering
    \centering
    \includegraphics[width=0.8\textwidth]{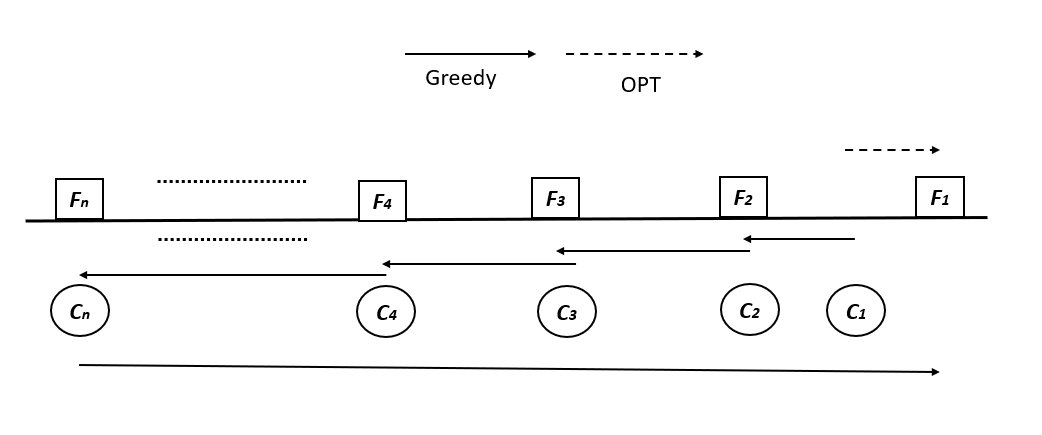} 
    \caption{\centering  Illustrating facility assignments with $n$ customers.}
    \label{fig:line3}
\end{figure}

\begin{center}
$\text{Greedy} = \frac{d}{2} + \frac{(n-1)n}{2} \cdot d = \frac{d}{2}\{1+(n-1)n\}$ \\[1em]
$\text{OPT} = \frac{d}{2}$ \\[1em]
$\frac{\text{Greedy}}{\text{OPT}} = 1+n(n-1) = n^2-n+1$
\end{center}

This formula captures the cumulative effect of each additional facility on the worst-case behavior. Since the optimal algorithm maintains its performance with cost $\frac{d}{2}$, the general competitive ratio becomes $n^2-n+1$.

\vspace{0.5cm}
\textbf{Case 4:} 

Changing the distance between each facility causes the distance to increase exponentially. 

\begin{figure}[htbp]
  \centering
    \centering
    \includegraphics[width=0.8\textwidth]{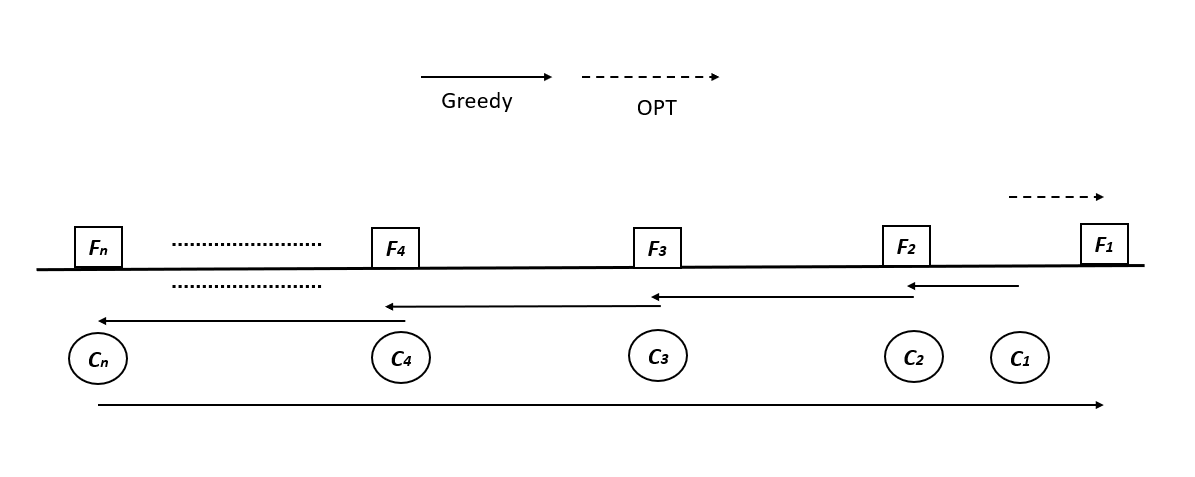} 
    \caption{\centering  Illustrating facility assignments with $n$ customers increase distance exponentially. }
    \label{fig:line4}
\end{figure}

\begin{center}
$Greedy = \frac{1}{2}d+d+2d+4d+\cdots\approx 2^n$\\
$OPT = \frac{d}{2}$\\
$Ratio = \frac{Greedy}{OPT}$\\
$=2^{n-1}$
\end{center}

These analyses demonstrate that the performance of the greedy algorithm for circular facility assignment heavily depends on the number of facilities and their distribution along the circle. The equidistant case provides better-guaranteed performance bounds with linear growth in competitive ratio while allowing varying distances, which can lead to exponential degradation in performance. This understanding is crucial for implementing effective facility placement strategies in real-world applications where circular configurations are necessary.

The mathematical foundation established through these cases provides a comprehensive framework for understanding both the behavior and limitations of greedy assignment in circular configurations. The progression from simple cases with ratios of $2n-1$ when the distance between two adjacent facilities is the same and the competitive ratios are $n^2-n+1$ and $2^n - 1$  to the general case helps explain why maintaining uniform facility spacing and limiting the size of circular segments are crucial strategies in practical implementations. These insights can guide both theoretical analysis and practical implementation of facility assignment systems in circular configurations.

\section{Conclusions} \label{Conclusions}

In this paper, we presented a comprehensive survey of the diverse geometric configurations along with the various competitive ratios associated with greedy facility assignment algorithms over those different geometric spatial arrangements. From the simplest to the most complex geometric forms, there is a noticeable decline in performance and a progressive increase in difficulty. Each geometry helps to shed insight into the complexity inherent in the online assignments.

The triangular configuration competitive ratio of $5$ sets the ground for understanding how geometric constraints influence assignment decisions: Even with a minimum number of vertices, the greedy algorithm can be forced to make assignments that are much worse than optimal. In this, the competitive ratio of the rectangular configuration is $7$; it just illustrates how increasing one more vertex greatly increases the possibility of suboptimal assignments, especially the possibility of longer diagonal paths compared to edge assignments.

Extending our analysis to regular n-sided polygons, we came upon a linear relationship between the number of vertices and competitive ratio, establishing a bound of $2n-1$. This relation is essential in showing how the number of facility arrangements is directly related to the performance of algorithms. The analysis indicates that with the increase in the number of facilities, the potential gap between greedy and optimal assignments increases linearly; therefore, facility networks with large scales need extra strategies for reasonable performance guarantees.

The circular configuration, being the limiting case as n approaches infinity, showed some interesting patterns in the competitive ratios: from particular cases of ratios $5$ and $2n-1$ for a small number of facilities to the general case of $n^2-n+1$ for equidistant facilities and finally to $2^n - 1$ for non-equidistant distances, we could see how facility distribution greatly affects performance. These results stress the importance of keeping facilities equidistant from each other in practice.

Our results, therefore, have some key implications for real-world applications of online facility assignments. The analysis suggests that, when possible, one should settle for simpler geometries, as those offer stronger guarantees of performance. When this is not possible, the spacing of facilities must be as uniform as possible to avoid the worst case. The results also suggest that divide-and-conquer approaches may yield superior overall performance, dividing large facility networks into smaller, simpler geometric components.

Future work can be done on randomized algorithms regarding these geometric structures, on considering different capacity constraints, and on devising hybrid approaches that combine the simplicity of a greedy assignment with more sophisticated techniques of optimization. The integration of machine learning for the prediction of customer arrival patterns could also provide interesting avenues for research toward finding improved online assignments in practical settings.

%\ack{This research was undertaken as part of the ALADDIN
%(Autonomous Learning Agents for Decentralised Data and Information
%Networks) project and is jointly funded by BAE Systems and EPSRC
%(Engineering and Physical Research Council) strategic partnership
%(EP/C548051/1).}

%\nocite{*}

%\bibliographystyle{compj}
\bibliographystyle{plain}
\bibliography{zip/ref}

\end{document}